\documentclass[12pt]{article}%
\usepackage{amsmath}
\usepackage{graphicx}
\usepackage{wrapfig}
\usepackage{lscape}
\usepackage{rotating}
\usepackage{epstopdf}
\usepackage{amsfonts}
\usepackage{amssymb}
\usepackage{color}%
\newtheorem{theorem}{Theorem}

\newenvironment{proof}[1][Proof: ]{\textbf{#1}\it }

\usepackage{tikz}
\usetikzlibrary{shapes,arrows}

\begin{document}
\title{On estimation of the effect lag of predictors 
and prediction in functional linear model}
\author{Haiyan Liu, Georgios Aivaliotis, Jeanine Houwing-Duistermaat
\\Department of Statistics\\University of Leeds}
\maketitle

\begin{abstract}
We propose a functional linear model to predict a response using multiple functional and 
longitudinal predictors and to estimate the effect lags of predictors. 
The coefficient functions are written as the expansion of a basis system 
(e.g. functional principal components, splines), and the coefficients of 
the fixed basis functions are estimated via optimizing a penalization criterion. 
Then time lags are determined by simultaneously searching on a prior grid mesh 
based on minimization of prediction error criterion. 
Moreover, mathematical properties of the estimated parameters and predicted responses
are studied and performance of the method is evaluated by extensive simulations.
\end{abstract}

\section{Introduction}
Temporal (time stamped) data are collected both routinely and ad hoc for various processes related to human activities and the natural world. In its two extreme forms, this data can be sampled densely and regularly in time (we call this dense data) or can include only records obtained at irregular time intervals with few measurements (we call this sparse longitudinal data). Naturally, intermediate situations are also available. Examples of dense data are hourly pollution and climate measurements in a particular site, or financial time series. Sparse datasets can arise from medical data (e.g. visits to GP) and other ad hoc observations for example measurements on wild species to which access is not easy. 

Relationships between temporal data are often not synchronous and involve a delay in the effects. For example, historical exposure to high temperatures might not have an effect on the growth of trees anymore after a certain period and it may also take some time before high temperatures result in lower growth rate.
It might take some time to have an effect on a person's health and similarly the effects might fade away after some time if the exposure to a factor seizes (e.g. stop smoking).

In this paper, we consider estimation and prediction in a functional regression model where the dense functional predictor trajectory and the sparse longitudinal predictor trajectory from certain intervals of past have effects on the sparse response trajectory. We estimate the intervals through the corresponding lags of the effect of predictors on response. In our motivating example, we estimate 
the influence of dense functional temperature on sparse longitudinal tree diameters. 
Moreover, we want to estimate the effect lags of temperature on tree diameter, i.e. from when the predictors have influence on the response and until when this influence disappears. 

The classical function-on-function linear model reads as follows :
$$Y(t)=\beta_0(t)+\int_0^T\beta_1(s,t)X(s)ds+\epsilon(t),\ t\in[0,T]$$
where $Y(t)$ is the response trajectory, $X(s)$ is the predictor trajectory, 
$\epsilon(t)$ is the error process, $\beta_0(t)$ is the intercept process, 
$\beta_1(s, t)$ is the two-dimensional regression coefficient function
which shows the influence of $X$ on $Y$. 
This model was first introduced by Ramsay and Dalzell (1991). 
For reviews of functional data analysis, see Ramsay and Silverman (2005), 
Horvath and Kokoszka(2012) and the references therein.
Notice that in this model the entire predictor trajectory $X(s)$ including the future values, 
i.e. when $s>t$, is assumed to influence the current value of response trajectory $Y$ at time $t$. 
Clearly this is not appropriate in many applications.

As a result, the historical functional linear model has been 
investigated by Malfait and Ramsay (2003), Harezlak et al. (2007), 
Kim et al. (2009, 2011)
where only the past of the predictor trajectory influences the response at the
current time:
$$Y(t)=\beta_0(t)+\int_{t-\delta_2}^{t-\delta_1}\beta_1(s,t)X(s)ds+\epsilon(s),\ t\in[0,T]$$
where $\delta_1$ and $\delta_2$ ($0<\delta_1<\delta_2<T$) are the lags for the influence 
of predictor trajectory on response trajectory. 
For one dense functional predictor, Malfait and Ramsay (2003) considers the 
triangular basis expansion of the coefficient function which is estimated at 
each observation point.
A penalized approach which allows varying lags for the historical 
functional linear model has been developed by Harezlak et al. (2007).
Kim et al. (2011) consider the situation that both predictor process
and response process are sparsely and irregularly observed.
Pomann et al. (2016) has extended the historical functional linear model to
multiple homogeneous predictors, and the response is influenced by
the predictors from a fixed starting effect time to current time.

The contribution of this paper is 
multiple heterogeneous (sparse longitudinal or dense functional) predictors 
are included, time lags (both starting and end points) that are fixed but unknown 
are determined, the asymptotic properties of the estimators have been investigated. 
To be precise, this paper addresses the historical functional linear model with multiple
heterogeneous predictors, and the response is influenced by predictors from
a fixed starting effect time to a fixed ending effect time. 
We estimate the coefficient functions, the effect lags and predict the
response. 
Moreover, the asymptotic behavior of the estimated coefficient functions,
and the predicted response curve is investigated.

The paper is organized as follows. 
In section 2, the history function-on-function linear model for multiple
heterogeneous predictors is introduced.
In section 3, we consider the estimation of the coefficient functions and
the uniform consistency of our estimators are established.
In section 4, the prediction of the response trajectories is proposed and the 
asymptotic property of the predicted trajectories is established.
The determination of the lags is proposed in section 5.
Extensive numerical examples are considered in section 6 to show
the finite properties of our proposed estimators.
In section 7, the Amazonian rainforest dataset is analysed and the lags are
determined.
We finish the paper with conclusion and discussion.

\section{Model}
Suppose our observations are $\{Y_{ij}, t_{ij}: i =1,...,n, \ j=1,...,m_{Yi} \}$,
 $\{W_{1ij}, s_{1ij}: i =1,...,n, \ j=1,...,m_{X1i} \}$
and $\{W_{2ij}, s_{2ij}: i=1,...,n, \ j=1,...,m_{X2i} \}$,
where $t_{ij},\ s_{1ij}, \ s_{2ij}\in[0, 1]$.
For example, the response $Y_{ij}$ corresponds to the tree diameter for subject $i$ at time $t_{ij}$.
The predictor $W_{1ij}$ corresponds to the temperature for subject $i$ at time $s_{1ij}$.
The predictor $W_{2ij}$ corresponds to the climatic water deficit for subject $i$ at time $s_{2ij}$.

Let $W_{1ij}=X_{1i}(s_{1ij})+\epsilon_{1ij}$, $W_{2ij}=X_{2i}(s_{2ij})+\epsilon_{2ij}$  
and $X_{1i}(t)$, $X_{2i}(t)$ are independent copies of underlying square-integrable 
random functions $X_1(t)$ and $X_2(t)$ over $[0, 1]$ respectively.
Without loss of generality, we assume $\mu_{X_1}(t)=E[X_{1}(t)]=0$ and 
$\mu_{X_2}(t)=E[X_{2}(t)]=0$.
We denote $C_{X_1}(s, t)=cov(X_1(s), X_1(t))$ the covariance of $X_1$ and
 $C_{X_2}(s, t)=cov(X_2(s), X_2(t))$ the covariance of $X_2$.
We assume that the first predictor curves $X_{1i}$ are observed on a dense 
and regular grid of points $s_{1ij}=s_{1j}$.
The observations $W_{1ij}$ are the discrete version of $X_{1i}$ with iid mean-zero 
and variance-finite noise $\epsilon_{1ij}$ which are independent of $X_{1i}$.
However, the second predictor curves $X_{2i}$ are observed on a sparse 
and irregular grid of points $s_{2ij}$.
Also observations $W_{2ij}$ are the discrete version of $X_{2i}$ with iid mean zero 
and variance-finite noise $\epsilon_{2ij}$ which are independent of $X_{2i}$.
For the responses $Y_{ij}$, they are observed on a sparse and irregular grid of points $t_{ij}$.

We define the lag historical functional linear model with two heterogeneous 
covariates $X_1$ and $X_2$ for the response $Y$ as 
\begin{align}
\label{Equation-LagFLM}
Y_{ij}= \beta_0(t_{ij})+\int_{\delta_{11}}^{\delta_{12}}\beta_1(s, t_{ij})X_{1i}(t_{ij}-s)ds
              +\int_{\delta_{21}}^{\delta_{22}}\beta_2(s, t_{ij})X_{2i}(t_{ij}-s)ds+e_{ij}
\end{align}
where $i\in\{1,...,n\}, \ j\in\{1,...,m_{Yi}\}, \ \beta_0: [0, 1]\to\mathbb R$, 
$\Delta_1=[\delta_{11}, \delta_{12}]\subset[0, 1]$,
$\Delta_2=[\delta_{21}, \delta_{22}]\subset[0, 1]$,
$\beta_1:\Delta_1\times[0, 1]\to\mathbb R$  and $\beta_2:\Delta_2\times[0, 1]\to\mathbb R$ 
are continuous two-dimensional coefficient functions, and $e_{ij}$ are independent measurement errors
with mean zero and finite variance $\sigma_e^2$.
Errors $e_{ij}$ are assumed to be independent of $X_{1i}$ and $X_{2i}$.

Notice that (\ref{Equation-LagFLM}) is equivalent to 
$$Y_{ij}=
\beta_0(t_{ij})+\int_{t_{ij}-\delta_{12}}^{t_{ij}-\delta_{11}}\beta_1(t_{ij}-s, t_{ij})X_{1i}(s)ds
               +\int_{t_{ij}-\delta_{22}}^{t_{ij}-\delta_{21}}\beta_2(t_{ij}-s, t_{ij})X_{2i}(s)ds
               +e_{ij}$$
then the model (\ref{Equation-LagFLM}) means that given 
the entire predictor curves $X_{1i}$ and $X_{2i}$,
the response for subject $i$ at time $t_{ij}$ is only affected by the values of $X_{1i}$
over time-window $[t_{ij}-\delta_{12}, t_{ij}-\delta_{11}]$ and by the values of $X_{2i}$
over time-window $[t_{ij}-\delta_{22}, t_{ij}-\delta_{21}]$.
That is, $t_{ij}-\delta_{12}$ is the starting effective time and $t_{ij}-\delta_{11}$ is
the ending effective time for $X_{1i}$ to have effect on $Y_i$ at time $t_{ij}$;
$t_{ij}-\delta_{22}$ is the starting effective time and $t_{ij}-\delta_{21}$ is
the ending effective time for $X_{2i}$ to have effect on $Y_i$ at time $t_{ij}$.
The coefficient functions $\beta_1$ and $\beta_2$, weigh the values $X_{1i}$ 
and $X_{2i}$ over the time-windows $[t_{ij}-\delta_{12}, t_{ij}-\delta_{11}]$ 
and $[t_{ij}-\delta_{22}, t_{ij}-\delta_{21}]$ respectively.
The coefficient functions $\beta_1$ and $\beta_2$ quantify the effect of $X_{1i}$ and $X_{2i}$  
respectively on the response $Y_{ij}$.

\section{Estimation}
Let $\{B_{11}(s), ..., B_{1K_1}(s)\}$ and $\{B_{21}(s), ..., B_{2K_2}\}_k$ 
be two pre-specified functional bases on 
$\Delta_1$ and $\Delta_2$. 
Then the two-dimensional coefficient functions $\beta_1(s, t)$ and $\beta_2(s, t)$ 
are assumed to be represented as
$$\beta_1(s, t)=\sum_{k=1}^{K_1}B_{1k}(s)b_{1k}(t),\ s\in\Delta_1,\ t\in[0, 1]$$ 
and
$$\beta_2(s, t)=\sum_{k=1}^{K_2}B_{2k}(s)b_{2k}(t),\ s\in\Delta_2,\ t\in[0, 1]$$
respectively, where $K_1$ and $K_2$ capture  the resolution of the fit and should be chosen accordingly
and $b_{1k}(t)$ and $b_{2k}(t)$ are the unknown time-varying coefficient functions defined on $[0, 1]$.
As Kim et al. (2011) reported where only one sparse predictor was discussed,
``the estimation is not sensitive to the choice of $K$ provided
that there are enough number of basis functions used in the estimation, since the
penalized solution (defined later in this session) prevents over-fitting''. 
Clearly, various basis functions such as Fourier, B-spline, wavelet basis can be used depending
on the specific features of the coefficient functions.
Since, we could not assume any prior on the coefficients and B-spline basis are computationally 
fast and have good properties, we use B-spline functions of degree 4 with 10 equally spaced interior 
knots over $\Delta_1$ and $\Delta_2$ (number of basis is 14).
For details on B-spline basis, see for example Fan and Gijbels (1996) and Ramsay and Silverman (2005).

Plugging $\beta_1(s, t)$ and $\beta_2(s, t)$ into equation (\ref{Equation-LagFLM}), we have
\begin{align}
\label{Equation-LagFLM-2}
Y_{ij}
= &\beta_0(t_{ij})
+\sum_{k=1}^{K_1}b_{1k}(t_{ij})\int_{\delta_{11}}^{\delta_{12}}B_{1k}(s)X_{1i}(t_{ij}-s)ds\\\nonumber
&\qquad \quad +\sum_{k=1}^{K_2}b_{2k}(t_{ij})\int_{\delta_{21}}^{\delta_{22}}B_{2k}(s)X_{2i}(t_{ij}-s)ds
+e_{ij}\\\nonumber
               =:&\beta_0(t_{ij})+\sum_{k=1}^{K_1}b_{1k}(t_{ij})\tilde{X}_{1ik}(t_{ij})
               +\sum_{k=1}^{K_2}b_{2k}(t_{ij})\tilde{X}_{2ik}(t_{ij})+e_{ij}\\\nonumber
               =:&\beta_0(t_{ij})+\mathbf b_1^T(t_{ij})\tilde{\mathbf X}_{1i}(t_{ij})
               +\mathbf b_2^T(t_{ij})\tilde{\mathbf X}_{2i}(t_{ij})+e_{ij}, \nonumber
\end{align}
where $\tilde{X}_{1ik}(t_{ij})=\int_{\delta_{11}}^{\delta_{12}}B_{1k}(s)X_{1i}(t_{ij}-s)ds$, 
$\tilde{X}_{2ik}(t_{ij})=\int_{\delta_{21}}^{\delta_{22}}B_{2k}(s)X_{2i}(t_{ij}-s)ds$,
$\mathbf b_1(t_{ij})=(b_{11}(t_{ij}),..., b_{1K_1}(t_{ij}))^T$,
$\mathbf b_2(t_{ij})=(b_{21}(t_{ij}),..., b_{2K_2}(t_{ij}))^T$,
$\tilde{\mathbf X}_{1i}(t_{ij})=(\tilde X_{1i1}(t_{ij}), ..., \tilde X_{1iK_1}(t_{ij}))^T$,
and $\tilde{\mathbf X}_{2i}(t_{ij})=(\tilde X_{2i1}(t_{ij}), ..., \tilde X_{2iK_2}(t_{ij}))^T$.
Note the observed times $t_{ij}$ depend on subject $i$.
Then model (\ref{Equation-LagFLM}) reduces to a varying coefficient model
with $K_1$ induced predictors $\tilde{X}_{1ik}(t_{ij})$ 
and $K_2$ induced predictors  $\tilde{X}_{2ik}(t_{ij})$. 

At first, notice that $\mu_{X_1}(t)=\mu_{X_2}(t)=0$ implies $\beta_0(t_{ij})=E[Y_{ij}]$, so $\beta_0$ can be estimated
by smoothing $Y_{ij}$ via local smoothing method based on the pooled data, see for example
Yao et al. (2005), Beran and Liu (2014) and Liu and Houwing-Duistermaat (2018).
We denote $Y_{ij}-\hat\beta_0(t_{ij})$ by $Y_{ij}$, where $\hat\beta_0(t_{ij})$ is an estimator of $\beta_0(t)$ evaluated at time $t_{ij}$.

In order to derive the estimator of $\{b_{11}(t),...,b_{1K_1}(t)\}$ and
$\{b_{21}(t),...,b_{2K_1}(t)\} $, we assume $t_{ij}=t_j^0$ only in this
paragraph, i.e. the 
observation times for different subject are the same.
We then estimate $b_{1k}(t_j^0)$ and $b_{2k}(t_j^0)$ by minimizing:
\begin{align}
\label{Equation-PSSE}
PSSE_{b_1, b_2}
=&\sum_{i=1}^n e_{ij}^2
   +\rho_1\|\mathbf b_1(t_j^0)\|^2
  +\rho_2\|\mathbf b_2(t_j^0)\|^2\\\nonumber
=&\|\mathbf Y_{j}-\mathbf b_1^T(t_j^0)\tilde{\mathbf X}_{1i}(t_j^0)
   +\mathbf b_2^T(t_j^0)\tilde{\mathbf X}_{2i}(t_j^0)\|^2
  +\rho_1\|\mathbf b_1(t_j^0)\|^2
  +\rho_2\|\mathbf b_2(t_j^0)\|^2\nonumber
\end{align}
where $\|\cdot\|$ is the Euclidean norm of a vector, 
$\mathbf Y_{j}=(Y_{1j}, ..., Y_{nj})^T$,
$\rho_1>0$ and $\rho_2>0$ are the regularization parameters which are assumed to be constants
for any time $t\in[0, 1]$ in order to reduce the high variability if they vary for each time.
The penalization does not only prevent over-fitting but also guarantee the inverse of matrix
while solving the minimization problem.
Then the minimizer of (\ref {Equation-PSSE}) is 
\begin{align*}
  \begin{bmatrix}
   \hat{ \mathbf b}_1(t_j^0) \\
    \hat{ \mathbf b}_2(t_j^0)
  \end{bmatrix}
  =\left(
  \mathbf Z_{j}^T\mathbf Z_{j}
  +  \begin{bmatrix}
   \rho_1I_{K_1} & 0 \\
   0 &  \rho_2I_{K_2} 
  \end{bmatrix}
  \right)^{-1}
  \left(
  \mathbf Z_{j}^T\mathbf Y_{j}
  \right)
\end{align*}
where $I_K$ is the $K\times K$ identity matrix and
\begin{align*}
\mathbf Z_{j}
=\begin{bmatrix}
   \tilde X_{111}(t_j^0)&\cdots & \tilde X_{11K_1}(t_j^0)
   &\tilde X_{211}(t_j^0)&\cdots & \tilde X_{21K_2}(t_j^0) \\
   \vdots & & \vdots&\vdots& &\vdots\\
   \tilde X_{1n1}(t_j^0)&\cdots & \tilde X_{1nK_1}(t_j^0)
   &\tilde X_{2n1}(t_j^0)&\cdots & \tilde X_{2nK_2}(t_j^0) 
  \end{bmatrix}.
  \end{align*}
  
Therefore, by using the probability limits of the covariance structure, for arbitrary $t\in[0, 1]$, we have
\begin{align}
\label{Equation-Solution}
  \begin{bmatrix}
   \hat{ \mathbf b}_1(t) \\
    \hat{ \mathbf b}_2(t)
  \end{bmatrix}
  =\left(
 \begin{bmatrix}
  \hat{\mathbf C}_{11}(t)  &  \hat{\mathbf C}_{12}(t) \\
  \hat{\mathbf C}_{21}(t)  & \hat{\mathbf C}_{22}(t)
  \end{bmatrix}
  +  \begin{bmatrix}
   \frac{\rho_1}{n}I_{K_1} & 0 \\
   0 &  \frac{\rho_2}{n}I_{K_2} 
  \end{bmatrix}
  \right)^{-1}
 \begin{bmatrix}
 \hat{\mathbf C}_{1Y}(t) \\
  \hat{\mathbf C}_{2Y}(t) 
  \end{bmatrix}
\end{align}
where 
$ \hat{\mathbf C}_{11}(t)=\left[\hat C_{\tilde{X}_{1k}, \tilde{X}_{1l}}(t)\right]_{kl}$ is a $K_1\times K_1$  matrix
with $\hat C_{\tilde{X}_{1k}, \tilde{X}_{1l}}(t)$ an estimator of 
$C_{\tilde{X}_{1k}, \tilde{X}_{1l}}(t)=cov\left(\tilde{X}_{1k}(t), \tilde{X}_{1l}(t)\right)$, 
$ \hat{\mathbf C}_{12}(t)=\left[\hat C_{\tilde{X}_{1k}, \tilde{X}_{2l}}(t)\right]_{kl}$ is a $K_1\times K_2$  matrix
with $\hat C_{\tilde{X}_{1k}, \tilde{X}_{2l}}(t)$ an estimator of 
$C_{\tilde{X}_{1k}, \tilde{X}_{2l}}(t)=cov\left(\tilde{X}_{1k}(t), \tilde{X}_{2l}(t)\right)$, 
$ \hat{\mathbf C}_{21}(t)=\left[\hat C_{\tilde{X}_{2k}, \tilde{X}_{1l}}(t)\right]_{kl}$ is a $K_2\times K_1$  matrix
with $\hat C_{\tilde{X}_{2k}, \tilde{X}_{1l}}(t)$ an estimator of 
$C_{\tilde{X}_{2k}, \tilde{X}_{1l}}(t)=cov\left(\tilde{X}_{2k}(t), \tilde{X}_{1l}(t)\right)$, 
$ \hat{\mathbf C}_{22}(t)=\left[\hat C_{\tilde{X}_{2k}, \tilde{X}_{2l}}(t)\right]_{kl}$ is a $K_2\times K_2$  matrix
with $\hat C_{\tilde{X}_{2k}, \tilde{X}_{2l}}(t)$ an estimator of 
$C_{\tilde{X}_{2k}, \tilde{X}_{2l}}(t)=cov\left(\tilde{X}_{2k}(t), \tilde{X}_{2l}(t)\right)$, 
$\hat{\mathbf C}_{1Y}(t)=\left[\hat C_{\tilde X_{11}, Y}(t), ..., \hat C_{\tilde X_{1K_1}, Y}(t)\right]^T$ 
is a vector and $\hat C_{\tilde X_{1l}, Y}(t)$ is estimator of 
$C_{\tilde X_{1l}, Y}(t)=cov\left(\tilde{X}_{1l}(t), Y(t)\right)$, and
$\hat{\mathbf C}_{2Y}(t)=\left[C_{\tilde X_{21}, Y}(t), ..., C_{\tilde X_{2K_2}, Y}(t)\right]^T$ 
is a vector and $\hat C_{\tilde X_{2l}, Y}$ is an estimator of 
$C_{\tilde X_{2l}, Y}(t)=cov\left(\tilde{X}_{2l}(t), Y(t)\right)$.

To obtain the necessary quantities in (\ref {Equation-Solution}), we consider the covariances:
\begin{itemize}
\item For $C_{\tilde{X}_{1k}, \tilde{X}_{1l}}(t)$, we have 
\begin{align*}
C_{\tilde{X}_{1k}, \tilde{X}_{1l}}(t)
&=cov\left(\tilde{X}_{1k}(t), \tilde{X}_{1l}(t)\right)\\
&=\int_{\delta_{11}}^{\delta_{12}}\int_{\delta_{11}}^{\delta_{12}}
   B_{1k}(s)B_{1l}(u)E[X_1(t-s)X_1(t-u)]duds\\
&=\int_{\delta_{11}}^{\delta_{12}}\int_{\delta_{11}}^{\delta_{12}}
   B_{1k}(s)B_{1l}(u)C_{X_1}(t-s, t-u)duds
\end{align*}
where $C_{X_1}(s, u)$ is the covariance between $X_1(s)$ and $X_1(u)$.
Since predictor $X_1$ is densely observed, $C_{X_1}(s, u)$ can be 
estimated by bivariate kernel smoothing, see Beran and Liu (2014):
\begin{align*}
\hat C_{X_1}(s, u)= 
\frac{1}{(m_{X1}b)^2}\sum_{j,k=1}^{m_{X1}}
K\left(\frac{s-s_{1j}}{b}, \frac{u-s_{1k}}{b}\right)
\frac1n\sum_{i=1}^n W_{1ij}W_{1ik}
\end{align*}
where $b$ is a bandwidth and $K$ is a bivariate kernel function.

\item For $C_{\tilde{X}_{2k}, \tilde{X}_{2l}}(t)$, we have 
\begin{align*}
C_{\tilde{X}_{2k}, \tilde{X}_{2l}}(t)
&=cov\left(\tilde{X}_{2k}(t), \tilde{X}_{2l}(t)\right)\\
&=\int_{\delta_{21}}^{\delta_{22}}\int_{\delta_{21}}^{\delta_{22}}
   B_{2k}(s)B_{2l}(u)E[X_2(t-s)X_2(t-u)]duds\\
&=\int_{\delta_{21}}^{\delta_{22}}\int_{\delta_{21}}^{\delta_{22}}
   B_{2k}(s)B_{2l}(u)C_{X_2}(t-s, t-u)duds
\end{align*}
where $C_{X_2}(s, u)$ is the covariance between $X_2(s)$ and $X_2(u)$.
Since predictor $X_2$ is sparsely observed, $C_{X_2}(s, u)$ can be 
estimated by local linear surface smoother (Yao et al. 2015) 
which is defined through minimizing
\begin{align*}
\sum_{i=1}^n \frac{1}{(m_{X2i}b)^2}\sum_{j\ne k=1}^{m_{X2i}}
K\left(\frac{s-s_{2ij}}{b}, \frac{u-s_{2ik}}{b}\right)
(W_{2ij}W_{2ik}-\alpha_0-\alpha_1(s-s_{2ij})-\alpha_2(u-s_{2ik}))^2
\end{align*}
with respect to $\alpha_0,\ \alpha_1,\ \alpha_2$,
where $b$ is a bandwidth and $K$ is a bivariate kernel function. 
And $\hat C_{X_2}(s, u)=\hat \alpha_0$.

\item For $C_{\tilde{X}_{1k}, \tilde{X}_{2l}}(t)$, we have 
\begin{align*}
C_{\tilde{X}_{1k}, \tilde{X}_{2l}}(t)
&=cov\left(\tilde{X}_{1k}(t), \tilde{X}_{2l}(t)\right)\\
&=\int_{\delta_{11}}^{\delta_{12}}\int_{\delta_{21}}^{\delta_{22}}
   B_{1k}(s)B_{2l}(u)E[X_1(t-s)X_2(t-u)]duds\\
&=\int_{\delta_{11}}^{\delta_{12}}\int_{\delta_{21}}^{\delta_{22}}
   B_{1k}(s)B_{2l}(u)C_{X_1, X_2}(t-s, t-u)duds
\end{align*}
where $C_{X_1, X_2}(s, u)$ is the covariance between $X_1(s)$ and $X_2(u)$.
Since predictor $X_1$ is densely observed and $X_2$ is sparsely observed, 
$C_{X_1}(s, u)$ can be estimated by local surface smoothing.

\item For $C_{\tilde{X}_{2k}, \tilde{X}_{1l}}(t)$, it is similar to 
$C_{\tilde{X}_{1k}, \tilde{X}_{2l}}(t)$.

\item For $C_{\tilde{X}_{1l}, Y}(t)$, we have 
\begin{align*}
C_{\tilde{X}_{1l}, Y}(t)
&=cov\left(\tilde{X}_{1l}(t), Y(t)\right)\\
&=\int_{\delta_{11}}^{\delta_{12}} B_{1l}(s)E[X_1(t-s)Y(t)]ds\\
&=\int_{\delta_{11}}^{\delta_{12}} B_{1l}(s)C_{X_1, Y}(t-s, t)ds
\end{align*}
where $C_{X_1, Y}(s, u)$ is the covariance between $X_1(s)$ and $Y(u)$.
Since $X_1$ is densely observed and $Y$ is sparsely observed, 
$C_{X_1, Y}(s, u)$ can be estimated by local linear surface smoothing.
\item For $C_{\tilde{X}_{2l}, Y}(t)$, it is similar to $C_{\tilde{X}_{1l}, Y}(t)$.
\end{itemize}

Once $\mathbf {\hat b}_1(t)$ and $\mathbf {\hat b}_2(t)$ are obtained 
(for given lags $\delta$'s and regularization 
parameters $\rho$'s), we can estimate coefficient functions by
$$\hat\beta_1(s, t)=\sum_{k=1}^{K_1}B_{1k}(s)\hat b_{1k}(t),\ s\in\Delta_1,\ t\in[0, 1]$$
and 
$$\hat\beta_2(s, t)=\sum_{k=1}^{K_2}B_{2k}(s)\hat b_{2k}(t),\ s\in\Delta_2,\ t\in[0, 1].$$
\begin{theorem}
Under assumptions in Beran and Liu (2014) and Yao et al. (2005a, 2005b), 
denote $I_t=[\max \{\delta_{12}, \delta_{22}\}, 1]$, 
$$\lim_{n\to\infty}\sup_{s, t\in\Delta_1\times I_t}|\hat\beta_1(s, t)-\beta_1(s, t)|
=0\quad {\text{in probability}}$$
$$\lim_{n\to\infty}\sup_{s, t\in\Delta_2\times I_t}|\hat\beta_2(s, t)-\beta_2(s, t)|
=0\quad {\text{in probability}}$$
\end{theorem}
\begin{proof}
Uniform consistency of $\hat C_{X_1}(s, u)$ is given in Theorem 4 of Beran and Liu (2014),
uniform consistency of $\hat C_{X_1,X_2},\ \hat C_{X_2,X_1},\ \hat C_{X_1,Y},
\ \hat C_{X_2,Y}$ 
is given in Lemma 1 of Yao et al. (2005b),
uniform consistency of $\hat C_{X_2}(s, u)$ is given in Theorem 1 of Yao et al. (2005a).
Then the uniform consistency of $\mathbf {\hat C}_{11}(t),\ \mathbf {\hat C}_{12}(t),
\ \mathbf {\hat C}_{21}(t),\ \mathbf {\hat C}_{22}(t),\ \mathbf {\hat C}_{1Y}(t),
\ \mathbf {\hat C}_{2Y}(t)$ 
can be obtained.
Therefore the uniform consistency of $\hat {\mathbf b}_1(t)$ and $\hat {\mathbf b}_2(t)$
follows and thus that of $\hat\beta_1(s, t)$ and $\hat\beta_2(s, t)$ can be obtained.
\end{proof}

\section{Prediction}
Suppose we observe a new discrete response curve $\mathbf Y^*_j=(Y^*(t_1^*),...Y^*(t_{m^*}^*))$,  
discrete dense predictor trajectory $\mathbf W_1^*=(W_1^*(s_{11}), ..., W_1^*(s_{1m_{X1}}))^T$ 
and discrete sparse predictor trajectory $\mathbf W_2^*=(W_2^*(s_{21}^*),...W_2^*(s_{2m_{X2}^*}^*))^T$.
From the original model (\ref{Equation-LagFLM}), the predicted response curve is
\begin{align}
\label{Equation-prediction1}
E\left[Y^*(t)|X_1^*, X_2^*\right]
=\beta_0(t)+\int_{\delta_{11}}^{\delta_{12}}\beta_1(s,t)X_1^*(t-s)ds
+\int_{\delta_{21}}^{\delta_{22}}\beta_2(s,t)X_2^*(t-s)ds.
\end{align}
However, the lags $\delta_{11}, \delta_{12}, \delta_{21}, \delta_{22}$ and regularization 
parameters $\rho_1$ and $\rho_2$ have to be determined
and the functional representation of the predictor trajectories $X_1^*(s)$
and $X_2^*(s)$ have to be recovered from data.

For $X_1^*(s)$, it can be easily recovered by kernel smoothing, since the sampling 
is dense. 

However for $X_2^*(s)$, since the sampling is sparse and irregular, 
we use functional principal component analysis (FPCA).
As discussed, we assume $X_2^*(s)\sim X_2(s)\in L^2[0, 1]$
and $E[X_2(s)]=0$.
Denote the covariance of $X_2(s)$ by $C_{X_2}(s, u)=cov(X_2(s), X_2(u))$, then
the Mercer's theorem gives the following spectral decomposition of the covariance
$$C_{X_2}(s, t)=\sum_{l=1}^\infty \lambda_l\phi_l(s)\phi_l(u)$$
where $\lambda_1\geq\lambda_2\geq...\geq0$ are eigenvalues and $\phi_l$ are orthonormal eigenfunctions.
By KL expansion, $X_2^*(s)$ can be represented as
$$X_2^*(s)=\sum_{l=1}^\infty\xi_l^*\phi_l(s)$$
where $\xi_l^*=\int_0^1 X_2^*(s)\phi_l(s)ds$ are the functional principal component scores and
are uncorrelated random variables with mean 0 and variance $\lambda_l$.
In practice, $X_2^*(s)$ is often truncated by only including the first several items, i.e.
$$X_2^{*L}(s)=\sum_{l=1}^L\xi_l^*\phi_l(s).$$
The covariance $C_{X_2}(s, t)$ can be estimated as we discussed in last section
and the eigenfunctions $\phi_l$ can be estimated following the spectral decomposition
of the estimated covariance.
However the scores $\xi_l^*$ cannot be approximated by numerical integration as 
we usually do for dense functional data.
In fact, under the Gaussian assumption, denote 
$\boldsymbol\phi_l=(\phi_l(s_{21}^*),...,\phi_l(s_{2m_{X_2}^*}^*))^T$,
the best linear predictor for $\xi_l^*$ is 
(see  Mardia et al. 1978, Yao et al. 2005 or see the application in Liu et al. 2018):
$$\tilde \xi_l^*=\lambda_l\boldsymbol\phi_l^T\Sigma^{-1}\mathbf W_2^*$$
where
$\Sigma=var(\mathbf W_2^*)$.
Then the estimate of $\xi_l^*$ can be defined as
$$\hat \xi_l^*=\hat \lambda_l\hat {\boldsymbol\phi_l}^T\hat \Sigma^{-1}\mathbf W_2^*.$$

The number of eigenfunctions $L$ can be selected to be the number of eigenfunctions that
explain 95\% of the functional covariance.
Once obtaining the estimation of eigenfunctions $\phi_l$, scores  $\xi_{il}$ and $L$, 
$X_2^*(s)$ can be recovered as
$$\hat X_2^*(s)=\sum_{l=1}^L\hat \xi_l^*\hat\phi_l^*(s).$$
After plugging the functional representation of the predictor curves 
$\hat X_1^*(s)$ and $\hat X_2^*(s)$ into (\ref{Equation-prediction1}), we have 
\begin{align}
\label{Equation-prediction2}
\hat Y_L^*(t)
&=\int_{\delta_{11}}^{\delta_{12}}\hat\beta_1(s,t)\hat X_1^*(t-s)ds
+\int_{\delta_{21}}^{\delta_{22}}\hat\beta_2(s,t)\hat X_2^*(t-s)ds \nonumber\\ 
&=\int_{\delta_{11}}^{\delta_{12}}\hat\beta_1(s,t)\hat X_1^*(t-s)ds
+\int_{\delta_{21}}^{\delta_{22}}\hat\beta_2(s,t)
\sum_{l=1}^L\hat\xi_l^*\hat\phi_l(t-s)ds.
\end{align}

Define
\begin{align*}
\tilde Y^*(t)
&=\int_{\delta_{11}}^{\delta_{12}}\beta_1(s,t)X_1^*(t-s)ds
+\int_{\delta_{21}}^{\delta_{22}}\beta_2(s,t)
\sum_{l=1}^\infty\tilde\xi_l^*\phi_l(t-s)ds.
\end{align*}
and 
\begin{align*}
\tilde Y_L^*(t)
&=\int_{\delta_{11}}^{\delta_{12}}\beta_1(s,t)X_1^*(t-s)ds
+\int_{\delta_{21}}^{\delta_{22}}\beta_2(s,t)
\sum_{l=1}^L\tilde\xi_l^*\phi_l(t-s)ds.
\end{align*}

\begin{theorem}
Under assumptions in Beran and Liu (2014) and Yao et al. (2005a, 2005b), 
denote $I_t=[\max \{\delta_{12}, \delta_{22}\}, 1]$, 
for all $t\in I_t$, we have
$$\lim_{n\to\infty}\hat Y_L^*(t)=\tilde Y^*(t)\quad \text{in probabilty}.$$
\end{theorem}
\begin{proof}
For fixed $L$, 
we have
\begin{align*}
&|\hat Y_L^*(t)-\tilde Y^*(t)|\\
&\leq |\hat Y_L^*(t)-\tilde Y_L^*(t)|+|\tilde Y_L^*(t)-\tilde Y^*(t)|\\
&\leq
\left|\int_{\delta_{11}}^{\delta_{12}}\hat\beta_1(s,t)\hat X_1^*(t-s)ds
-\int_{\delta_{11}}^{\delta_{12}}\beta_1(s,t)X_1^*(t-s)ds
\right|\\
&\quad +
\left|\int_{\delta_{21}}^{\delta_{22}}\hat\beta_2(s,t)
\sum_{l=1}^L\hat\xi_l^*\hat\phi_l(t-s)ds
-\int_{\delta_{21}}^{\delta_{22}}\beta_2(s,t)
\sum_{l=1}^L\tilde\xi_l^*\phi_l(t-s)ds
\right|\\
&\quad +
\left|\int_{\delta_{21}}^{\delta_{22}}\beta_2(s,t)
\sum_{l=1}^L\tilde\xi_l^*\phi_l(t-s)ds
-\int_{\delta_{21}}^{\delta_{22}}\beta_2(s,t)
\sum_{l=1}^\infty\tilde\xi_l^*\phi_l(t-s)ds
\right|\\
&=I_1+I_2+I_3
\end{align*}
For $I_1$, from the uniform consistency of $\hat\beta_1(s, t)$ established in Theorem 1
and the uniform consistency of kernel smoother, we have $I_1\to 0$ as $n\to\infty$.

For $I_2$, from the uniform consistency of $\hat\beta_2(s, t)$ established in Theorem 1,
the uniform consistency of $\hat\xi_l^*$ for $\tilde\xi_l^*$ from Theorem 3 in Yao et al.
(2005a), and the uniform consistency of $\hat\phi_l$ from Theorem 2 in Yao et al. (2005a),
we have $I_2\to0$ as $n\to\infty$.

For $I_3$, following Lemma A.3 in Yao et al. (2005a), we have $I_3\to 0$ as $n\to \infty$.

Therefore, Theorem 2 follows.
\end{proof}

\section{Implementation}
The final question is to estimate the time lag $\delta$'s which is of great importance in our application.
For selecting $\delta$'s and $\rho$'s, we consider the Normalized Prediction Error (NPE) criterion 
and the $K$-fold cross validation criterion.
Specifically, NPE in this situation is defined as
\begin{align}
NPE\{(\delta, \lambda)\}
=\frac{1}{N}\sum_{i=1}^n\sum_{j=1}^{m_{Yi}}\frac{\left|\hat Y_{ij}-Y_{ij}\right|}{|Y_{ij}|}
\end{align}
where $\hat Y_{ij}$ is the predicted value for the $j$th measurement on the $i$th response
trajectory $Y(t)$ obtained using $\delta$'s and $\lambda$'s, $N=\sum_{i=1}^nm_{Yi}$.
Divide the data into $K$ equal parts, for each $k=1, ..., K$, fit the model with parameter
$\delta,\ \lambda$ to the other $K-1$ parts, giving the estimation of coefficient functions,
further giving the prediction $\hat Y^{-k}_{ij}$ in the $k$th part, 
and then compute the prediction error in the $k$th part.
The $K$-fold cross validation score is defined as,
\begin{align}
CV\{(\delta, \lambda)\}
=\frac{1}{K}\sum_{k=1}^{K}\sum_{i\in k\text{th part} }\sum_{j=1}^{m_{Yi}}
\left(\hat Y^{-k}_{ij}-Y_{ij}\right)^2.
\end{align}
Similar criteria are considered in Kim et al. (2011) and Pomann et al. (2016).

Then $\delta$'s and $\rho$'s are chosen in a hierarchical manner.
Let $D_1$ and $D_2$ be the sets of potential lags for the first and second predictor, i.e.
$\{(\delta_{11}, \delta_{12})\}$ and $\{(\delta_{21}, \delta_{22})\}$, respectively.
Let $D_\rho$ be the sets of potential regularization parameters $\{(\rho_1, \rho_2)\}$.
Firstly, for a fixed point of 
$\delta^0=\left(\delta_{11}^0, \delta_{12}^0, \delta_{21}^0, \delta_{22}^0\right)\in D_1\times D_2$, 
NPE values are calculated for all $\rho\in D_\rho=\{(\rho_1, \rho_2)\}$.
Then the $\rho$ that achieves the smallest NPE value is chosen as the optimal 
$\rho$ for the given fixed point of lags $\delta^0$.
Secondly, The optimal $\rho$ is used for calculating the cross validation score for $\delta^0$.
At last, we repeat the above steps for all $\delta\in D_1\times D_2$ and the cross validation
score for all $\delta\in D_1\times D_2$ can be obtained. 
Then, the optimal $\delta$ is chosen to be the one with the smallest cross validation score.
Actually $D_1$ and $D_2$ are meshes in $[0, 1]$ and are chosen empirically, 
$D_\rho$ is also chosen empirically.

\section{Simulations}
We study efficiency of the NPE criterion for selecting the time lags $\delta$'s and 
regularization parameters $\rho$'s. 

For $n=50,\ 100,\ 150,\ 200$ subjects, we first generate the response curve $Y(t)$ and two
predictor curves $X_1(t)$ and $X_2(t)$ on a dense and equally spaced time points 
over $[0, 1]$, i.e. $\{j/99, j=0, ..., 99\}$.
The number of measurements made on the $i$th 
response $m_{Yi}$ is randomly selected from 20 to 50, 
the number of measurements made on the $i$th predictor 
$m_{Xi1}$ is 100 and the number of measurements made on the $i$th 
predictor $m_{Xi2}$ is randomly selected from 30 to 50.

Define $X_{1i}(t)=\xi_{i1}\sin(2\pi t)+\xi_{i2}t^2$ with
$\xi_{i1}\overset{iid}\sim N(0, 1)$ and $\xi_{i2}\overset{iid}\sim N(0, 1)$,
$X_{2i}(t)=\zeta_i\cos(2\pi t)$ with $\zeta_i\overset{iid}\sim N(0, 1)$.
We take the same time lags for both $X_1$ and $X_2$, 
i.e. $\delta_{11}=\delta_{21}=0.1$, $\delta_{12}=\delta_{22}=0.4$.
For coefficient functions, we take
$\beta_0(t)=t+t^{1/5}$,
$\beta_1(s, t)=\sin(2\pi s)\cos(\pi t),\ t\in[0, 1],\ s\in[0.2, 0.4]$,
$\beta_2(s, t)=\sin(4\pi s)\cos(2\pi t),\ t\in[0, 1],\ s\in[0.2, 0.4]$. 
The measurement errors are taken to be independent normal with 
signal to noise ratio 20 for the predictors and response.

Figure \ref{Data} shows the simulated data with $n=100$. 
\begin{figure}[h]
\centering
\includegraphics[width=\textwidth]{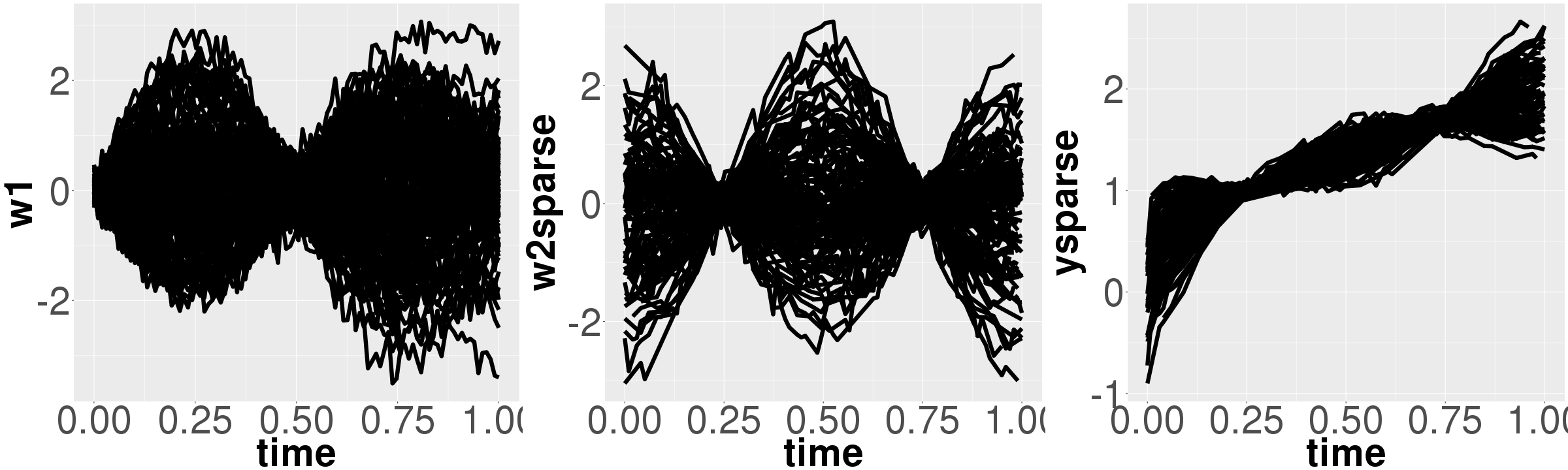} 
\caption{Simulation data: The left plot shows the discrete noisy observation of the first
predictor which is densely and regularly observed. 
The middle plot shows the discrete noisy observation of the second predictor which is
sparsely and irregularly observed.
The right plot shows the discrete noisy observation of the response which is also
sparsely and irregularly observed.}
\label{Data}
\end{figure}

The estimation is based on the B-spline 
(B-spline functions of degree 4 with 10 equally spaced interior 
knots over $[0,1]$) expansion of the coefficients. 
The number of functional principal components is chosen based on leave-one-curve 
cross validation criterion and 99\% variation is kept.
The penalized parameters  $\rho_1$ and $\rho_2$ are chosen on the dense grid 
of $\rho_1,\ \rho_2 \in [10^{-5}, 10^{-2}; 20]$.
We use NPE criterion and 10-fold cross validation criterion to determine 
the regularization parameters and the lags.
Notice that in order to check the estimation performance, the estimation procedure is
done under the correct lags, i.e. $\delta_{11}=\delta_{21}=0.1$, $\delta_{12}=\delta_{22}=0.4$.
Figure \ref{Estimate} shows the result of one simulation, where
$\rho_1$ is chosen as $4.28\times 10^{-4}$, $\rho_2$ is chosen as $8.86\times 10^{-4}$
and the corresponding NPE is $1.95\times 10^{-2}$. 
From Figure \ref{Estimate}, we conclude that our model successfully reveals the
structure of coefficient functions.

\begin{figure}[h]
\centering 
\includegraphics[width=\textwidth]{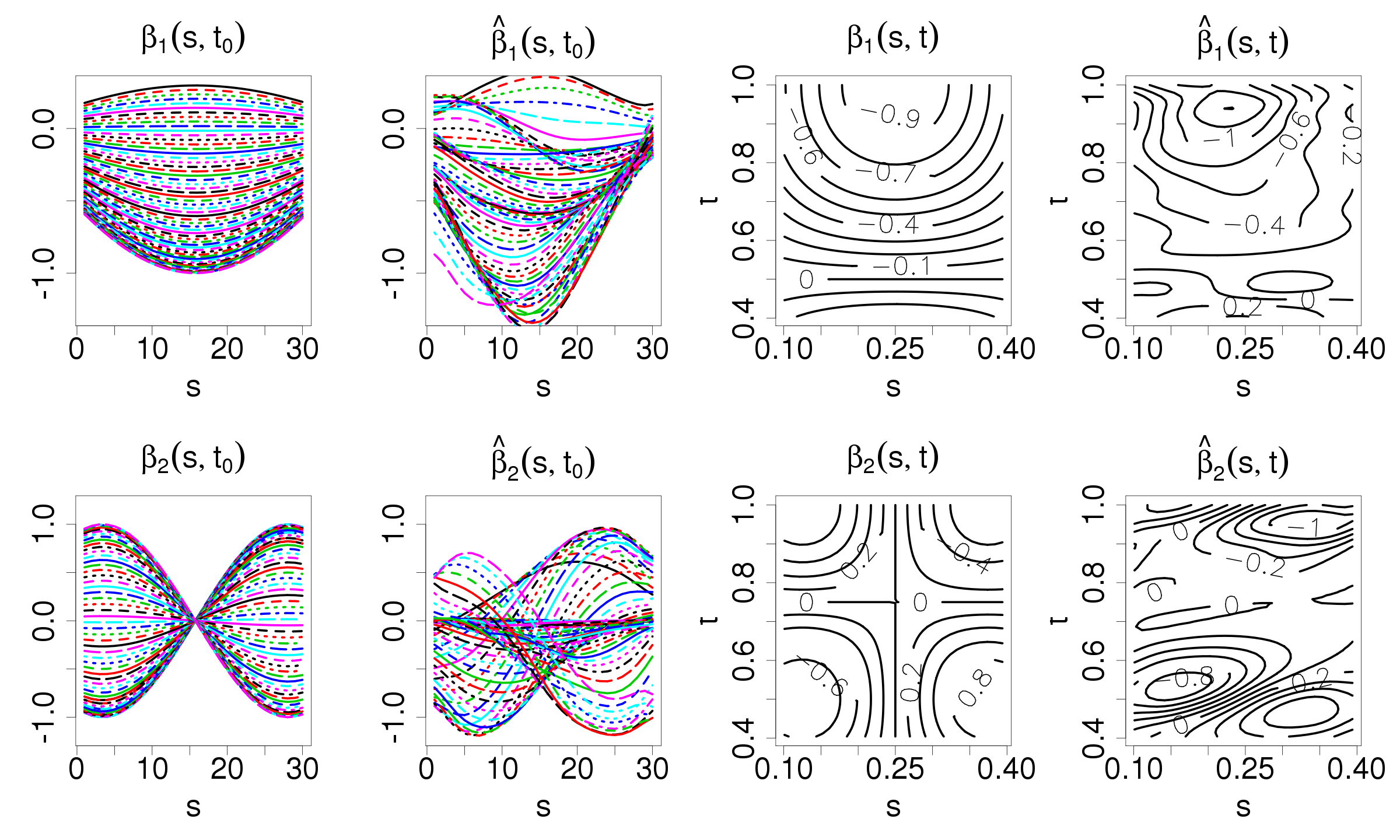}
\caption{One simulation result: The first left above plot is the true $\beta_1$:
abscissa is $s$ with the domain $[0.1, 0.4]$, 	
ordinate is the values of $\beta_1$,
there are 60 curves and they are $\beta_1(s, t_j)$ for $t_j=j/99,\ j=40,...,99$.
The second left above plot is the estimation of $\beta_1$.
The third left above plot is the contour line of the true $\beta_1$.
The fourth left above plot is the contour line of the estimated $\hat\beta_1$.
The bottom panel shows the true and estimated $\beta_2$.}
\label{Estimate}
\end{figure}

Table \ref{table1} shows the asymptotic properties of our estimation.
For different number of observations $n=50,\ 100,\ 150,\ 200$, the NPE
are shown and also the estimation is based on the correct lags.
As we can see, the NPE decrease as the $n$ increases which is correspond 
to the Theorem 1.

\begin{table}[ht]
\caption{NPEs based on correct lags}
\centering 
\begin{tabular}{|c|c c c c|}
\hline
n & 50 & 100 & 150 & 200 \\
\hline
NPE$\times 100$ & 2.08 & 1.95 & 1.86 & 1.79 \\
\hline
\end{tabular}
\label{table1}
\end{table}

For evaluating the performance of our model on selecting the effect lags,
the $\lambda$s are determined based on the NPE criterion and the $\delta$s
are determined based on 10-fold cross-validation score.
Since the true $\delta_{11}=\delta_{21}=0.1$ and $\delta_{21}=\delta_{22}=0.4$,
in order to save computational time, we fix the ending point 
i.e. $\delta_{11}=\delta_{21}=0.1$ 
and search the starting point $\delta_{21}=\delta_{22}\in\{0.3, 0.4, 0.5\}$.
That is we have three combinations but there is only one correct combination.
Our model has 65 correct choices out of 100 simulations.


\begin{thebibliography}{9}

\bibitem {BeranL14}Beran, J. and Liu , H. (2014).
On estimation of mean and covariance functions in repeated time series with
long-memory errors.
{\it Lithuanian Mathematical Journal}, 54(1), 8-34.

\bibitem{FanJ96} Fan, J. and Gijbels, I. (1996). 
{\it Local polynomial modeling and its applications}. 
CRC Press.



\bibitem{HarezlakLMC07}Harezlak, J., Coull, B. A., Laird, N. M., Magari, S. R., 
and Christiani, D. C. (2007). 
Penalized solutions to functional regression problems. 
{\it Computational statistics and data analysis}, 51(10), 4911-4925.


\bibitem{HorvathK12} Horvath, L. and Kokoszka, P. (2012). 
Inference for functional data with applications. 
Springer Science and Business Media.

\bibitem{KimSL11} Kim, K., Sent\"urk, D., and Li, R. (2011). 
Recent history functional linear models for sparse longitudinal data. 
{\it Journal of statistical planning and inference}, 141(4), 1554-1566.


\bibitem{LiuHD18} Liu, H. and Houwing-Duistermaat, J. (2018).
On trend and its derivative estimation in repeated unevenly spaced time series with
long-range dependent errors.
arXiv:1803.05411.

\bibitem{LiuDGHD18}  Liu, H., Del Galdo, F. and Houwing-Duistermaat, J. (2018).
Functional principal component analysis in predicting Scleroderma disease 
based on patients historical data. 


\bibitem{LopezLBP11} Lopez-Gonzalez, G., Lewis, S.L., Burkitt, M. and Phillips, O.L. (2011).
ForestPlots.net: a web application and research tool to manage and analyse tropical forest plot data. 
Journal of Vegetation Science 22: 610–613. doi: 10.1111/j.1654-1103.2011.01312.x        

\bibitem{LopezLBP09} Lopez-Gonzalez, G., Lewis, S.L., Burkitt, M., Baker T.R. and Phillips, O.L. (2009). 
ForestPlots.net Database. www.forestplots.net. Date of extraction [03,01,19].        		



\bibitem {MardiaKB79} Mardia, K. V., Kent, J. T., and Bibby, J. M. (1979). 
Multivariate Analysis. Academic Press.

\bibitem {MalfaitR03} Malfait, N. and Ramsay, J. O. (2003). 
The historical functional linear model. 
{\it Canadian Journal of Statistics}, 31(2), 115-128.

\bibitem {PengP09} Peng, J. and Paul, D. (2009). 
A geometric approach to maximum likelihood
estimation of the functional principal components from sparse longitudinal data. 
{\it Journal of Computational and Graphical Statistics}, 18(4),
995-1015.

\bibitem {PomannSL16} Pomann, G. M., Staicu, A. M., Lobaton, E. J., 
Mejia, A. F., Dewey, B. E., Reich, D. S., ... and Shinohara, R. T. (2016). 
A lag functional linear model for prediction of magnetization transfer ratio in multiple sclerosis lesions. 
{\it The Annals of Applied Statistics}, 10(4), 2325-2348.

\bibitem{RamsayD91} Ramsay, J. O. and Dalzell, C. J. (1991). 
Some tools for functional data analysis. 
{\it Journal of the Royal Statistical Society. Series B (Methodological)}, 539-572.

\bibitem {RamsayS05} Ramsay, J.O. and Silverman, B.W. (2005).
{\it Functional Data Analysis (Second Edition)}.

\bibitem {YaoMW05a} Yao, F., M\"uller, H. G. and Wang, J. L. (2005a). 
Functional data analysis for sparse longitudinal data. 
{\it Journal of the American Statistical Association}, 100(470), 577-590.

\bibitem {YaoMW05b} Yao, F., M\"uller, H. G., and Wang, J. L. (2005b). 
Functional linear regression analysis for longitudinal data. 
{\it The Annals of Statistics}, 33(6), 2873-2903.
\end{thebibliography}
\end{document}